\documentclass[conference]{IEEEtran}%
\usepackage{amsfonts}
\usepackage{amsmath}
\usepackage{amssymb}
\usepackage{amsthm}
\usepackage{mathtools}
\usepackage{graphicx}
\usepackage{hyperref}
\usepackage{xcolor}
\usepackage{soul}%
\providecommand{\U}[1]{\protect\rule{.1in}{.1in}}
\newtheorem{theorem}{Theorem}

\newtheorem{corollary}[theorem]{Corollary}

\newtheorem{definition}[theorem]{Definition}

\newtheorem{lemma}[theorem]{Lemma}

\newtheorem{remark}[theorem]{Remark}

\allowdisplaybreaks
\begin{document}

\title{A smallest computable entanglement monotone}

\author{\IEEEauthorblockN{Jens Eisert\IEEEauthorrefmark{1}
and Mark M.~Wilde\IEEEauthorrefmark{2}}
\IEEEauthorblockA{\IEEEauthorrefmark{1}%
Dahlem Center for Complex Quantum Systems,
Freie Universit\"at Berlin,
Berlin, Germany,\\
Helmholtz-Zentrum Berlin f{\"u}r Materialien und Energie, 
14109 Berlin, Germany,\\
Fraunhofer Heinrich Hertz Institute, 10587 Berlin, 
Germany,
Email: jenseisert@gmail.com}
\IEEEauthorblockA{\IEEEauthorrefmark{2}%
Hearne Institute for Theoretical Physics,
Department of Physics and Astronomy, \&
Center for Computation and Technology,\\
Louisiana State University,
Baton Rouge, Louisiana 70803, USA,
Email: mwilde@lsu.edu}}%
%

\maketitle
%

\begin{abstract}%

The Rains relative entropy of a bipartite quantum state is the tightest known upper bound on its distillable entanglement -- which has a crisp physical interpretation of entanglement as a resource -- and it is efficiently computable by convex programming. It has not been known to be a selective entanglement monotone in its own right. In this work, we strengthen the interpretation of the Rains relative entropy by showing that it is monotone under the action of selective operations that completely preserve the positivity of the partial transpose, reasonably quantifying entanglement. That is, we prove that Rains relative entropy of an ensemble generated by such an operation does not exceed the Rains relative entropy of the initial state in expectation, giving rise to the smallest, most conservative known computable selective entanglement monotone. Additionally, we show that this is true not only for the original Rains relative entropy, but also for Rains relative entropies derived from various R\'enyi relative entropies. As an application of these findings, we prove, in both the non-asymptotic and asymptotic settings, that the probabilistic approximate distillable entanglement of a state is bounded from above by various Rains relative entropies.



\end{abstract}%


\section{Introduction}

Entanglement is the feature  at the heart of quantum mechanics. One can
go as far to say it is the defining property that distinguishes it
from a classical theory \cite{BellBook}. In quantum information,
entanglement is seen as a resource for information-processing tasks 
\cite{HHHH09,PV07}.
For this reason, substantial efforts have been made to quantify and capture it, giving rise to a comprehensive theory of entanglement~\cite{HHHH09,PV07}.
%
Many measures of entanglement have been presented and their properties studied
over the years. Among those, one of the most prominent ones is the \emph{distillable
entanglement} \cite{BBPS96,BBPSSW96EPP,BDSW96}, which captures entanglement as a
resource: It quantifies the rate at which one can meaningfully extract Bell states, which can then be used, e.g., 
in quantum key distribution \cite{RevModPhys.74.145}.
Physically meaningful as this quantity is, it is
notoriously difficult to compute and known to be not convex \cite{SST01}. It may not even be Turing computable~\cite{SST01,WCP11,EMG12}.

Since the distillable entanglement is so important conceptually, and also
because a more pragmatic mindset has become more common over the years, researchers have resorted to using upper bounds on the
distillable entanglement.
The tightest known upper bound is the Rains
relative entropy \cite{R01}, a quantity that  interpolates the relative
entropy of entanglement \cite{VP98} and the logarithmic negativity \cite{ZHSL98,EP99,Vidal2002}. Due to this finding and since it can be computed as a convex program \cite{FWTD19,W18pra}, the Rains relative entropy is
considered a practically important bound
capturing notions of entanglement.

It has been well known since the original work \cite{R01} that the Rains relative entropy is non-increasing under the action of trace-preserving channels that completely preserve the \emph{positivity of the partial transpose (PPT)}, and as such, it is non-increasing under 
physically well motivated \emph{local operations and classical communication (LOCC)} channels that reflect the 
separated laboratory paradigm for bipartite quantum systems.
This is in fact one of the 
properties required to prove that it is an upper bound on distillable entanglement.
That said, importantly, what has been left open over the years
is to determine whether it actually has the stronger property of being a
monotone under selective quantum operations: This property
to be a \emph{selective entanglement monotone} 
\cite{VidalEntanglementMonotones,VP98}
is, however, an important property so that the Rains quantity can be 
regarded as a quantity meaningfully capturing entanglement according to this stronger notion.
This means that the average Rains relative entropy of an ensemble produced by 
suitable selective operations
does not exceed 
that
of the original state.

In this work, we bring this problem to rest by proving that the Rains relative entropy is a selective PPT monotone (thus implying that it is also a selective LOCC monotone). Since it is the tightest known upper bound on the distillable entanglement
-- which itself is the smallest reasonable entanglement monotone
\cite{PhysRevLett.84.2014} --
it is at the same time established as the smallest known computable measure of entanglement. Hence, it reasonably quantifies entanglement in a most conservative fashion, and as such provides a good 
guideline to assess quantum correlations in practically relevant 
settings.
Our proof has a general form, and so it applies not only to the original Rains relative entropy, but also to others constructed from various R\'enyi relative entropies.

As an application, we prove that the Rains relative entropy is an upper bound on the probabilistic approximate distillable entanglement of a bipartite state, in both the non-asymptotic and asymptotic settings. Since the probabilistic approximate distillable entanglement of a state is an upper bound on the usual distillable entanglement and it arguably has a more direct connection to experimental practice, our result strengthens the interpretation of the Rains relative entropy, thus giving the tightest known upper bound on the probabilistic approximate distillable entanglement of a state. 

%

\section{Preliminaries}

In this preliminary section, we set notation and review background material. In particular, we recall the definitions of various
relative entropies and their properties. These are then used to define Rains relative entropies, which are the signature entanglement measures considered
in this work. After that, we present some properties of the Rains relative
entropies and recall the definition of the set of completely
\emph{positive-partial-transpose (PPT) preserving channels}.

\subsection{Quantum states, channels, and relative entropies}

In what follows, naturally, notions of quantum states, channels 
and relative entropies feature strongly;
we point to the textbooks \cite{H06book,H12,W17,Wat18,KW20book} for 
more background on quantum information theory. 
We denote with  $\mathcal{S}(\mathcal{H})$ the
set of quantum states (unit-trace, positive semi-definite operators) acting on a Hilbert space $\mathcal{H}$. Let $\mathcal{L}%
(\mathcal{H})$ denote the set of linear operators acting on $\mathcal{H}$, and
let $\mathcal{L}_{+}(\mathcal{H})$ denote the set of positive semi-definite
operators acting on~$\mathcal{H}$. Let $\mathcal{N}:\mathcal{L}(\mathcal{H}%
)\rightarrow\mathcal{L}(\mathcal{H}^{\prime})$ denote a quantum channel, which
is a completely positive, trace-preserving map, and let $\operatorname{CPTP}%
(\mathcal{H},\mathcal{H}^{\prime})$ denote the set of all such quantum
channels taking $\mathcal{L}(\mathcal{H})$ to $\mathcal{L}(\mathcal{H}%
^{\prime})$. More generally, let $\operatorname{CP}(\mathcal{H},\mathcal{H}%
^{\prime})$ denote the set of completely positive maps taking $\mathcal{L}%
(\mathcal{H})$ to $\mathcal{L}(\mathcal{H}^{\prime})$.
The \emph{quantum relative entropy} of $\omega\in\mathcal{S}(\mathcal{H})$ and
$\tau\in\mathcal{L}_{+}(\mathcal{H})$ is defined by \cite{U62}
\begin{equation*}
D(\omega\Vert\tau)\coloneqq\left\{
\begin{array}
[c]{cc}%
\operatorname{Tr}[\omega(\log_{2}\omega-\log_{2}\tau)] & \text{if (F1)}\\
+\infty & \text{else}%
\end{array}
\right.  ,
\end{equation*}
where $\text{(F1)}\equiv\operatorname{supp}(\omega)\subseteq
\operatorname{supp}(\tau)$. The \emph{sandwiched R\'enyi relative entropy} is defined
for $\alpha\in(0,1)\cup(1,\infty)$ as 
\begin{align}
&  \widetilde{D}_{\alpha}(\omega\Vert\tau)\coloneqq\frac{1}{\alpha-1}\log
_{2}\widetilde{Q}_{\alpha}(\omega\Vert\tau),\label{eq:sandwiched-Renyi-def}\\
&  \widetilde{Q}_{\alpha}(\omega\Vert\tau)\coloneqq\left\{
\begin{array}
[c]{cc}%
\operatorname{Tr}\!\left[  \left(  \tau^{\frac{1-\alpha}{2\alpha}}\omega
\tau^{\frac{1-\alpha}{2\alpha}}\right)  ^{\alpha}\right]  & \text{if (F2)}\\
+\infty & \text{else}%
\end{array}
\right.  ,
\end{align}
where $\text{(F2)}\equiv\alpha\in(0,1)\vee(\alpha>1\wedge\operatorname{supp}%
(\omega)\subseteq\operatorname{supp}(\tau))$ \cite{MDSFT13,WWY13}. The following limit holds
\cite{MDSFT13,WWY13}%
\begin{equation*}
\lim_{\alpha\rightarrow1}\widetilde{D}_{\alpha}(\omega\Vert\tau)=D(\omega
\Vert\tau).
\end{equation*}
The following data-processing inequalities 
\begin{align}
D(\omega\Vert\tau)  &  \geq D(\mathcal{N}(\omega)\Vert\mathcal{N}%
(\tau)),\label{eq:DP}\,
\widetilde{D}_{\alpha}(\omega\Vert\tau)    \geq\widetilde{D}_{\alpha
}(\mathcal{N}(\omega)\Vert\mathcal{N}(\tau)), 
\end{align}
hold for $\omega\in\mathcal{S}%
(\mathcal{H})$, $\tau\in\mathcal{L}_{+}(\mathcal{H})$, $\mathcal{N}%
\in\operatorname{CPTP}(\mathcal{H},\mathcal{H}^{\prime})$, and $\alpha
\in\lbrack1/2,1)\cup(1,\infty)$,
with the first inequality established by 
ref.~\cite{Lindblad1975} and the second by
ref.~\cite{FL13} (see also
ref.~\cite{W17f}). The Petz and geometric R\'enyi relative
entropies are defined by the same construction in
\eqref{eq:sandwiched-Renyi-def}, but by replacing $\widetilde{Q}_{\alpha}$
with $Q_{\alpha}(\omega\Vert\tau)\coloneqq \operatorname{Tr}[\omega^{\alpha
}\tau^{1-\alpha}]$ \cite{P86}\ and $\widehat{Q}_{\alpha}(\omega\Vert
\tau)\coloneqq \operatorname{Tr}[\tau(\tau^{-1/2}\omega\tau^{-1/2})^{\alpha}]$
\cite{M13,Matsumoto2018,KW20book}, respectively. A
data-processing inequality, similar to 
the above,
holds for them
for $\alpha\in(0,1)\cup(1,2]$.
The following property of the quantum relative entropy is well known \cite{KW20book}.

\begin{lemma}[Direct-sum property]
\label{lem:direct-sum-rel-ent-prop}Let $\kappa\in\mathcal{S}(\mathcal{H}%
_{XA})$ and $\lambda\in\mathcal{L}_{+}(\mathcal{H}_{XA})$ be
classical--quantum, i.e., of the form%
\begin{equation*}
\kappa\coloneqq\sum_{x}p(x)|x\rangle\!\langle x|\otimes\kappa_{x},\quad
\lambda\coloneqq\sum_{x}q(x)|x\rangle\!\langle x|\otimes\lambda_{x},
\end{equation*}
where $\{|x\rangle\}_{x}$ is an orthonormal basis for $\mathcal{H}_{X}$,
$\{p(x)\}_{x}$ is a probability distribution, $\kappa_{x}\in\mathcal{S}%
(\mathcal{H}_{A})$ for all~$x$, $\{q(x)\}_{x}$ is a non-negative function, and
$\lambda_{x}\in\mathcal{L}_{+}(\mathcal{H}_{A})$ for all~$x$. Then%
\begin{equation*}
D(\kappa\Vert\lambda)=D(p\Vert q)+\sum_{x}p(x)D(\kappa_{x}\Vert\lambda_{x}),
\end{equation*}
where the classical relative entropy is defined as%
\begin{equation*}
D(p\Vert q)\coloneqq\left\{
\begin{array}
[c]{cc}%
\sum_{x}p(x)\log_{2}\!\left(  \frac{p(x)}{q(x)}\right)  & \text{if
}\operatorname{supp}(p)\subseteq\operatorname{supp}(q)\\
+\infty & \text{else.}%
\end{array}
\right.
\end{equation*}

\end{lemma}

A related property actually holds true also 
for sandwiched, Petz, and geometric R\'enyi
relative entropy. Building on ref.~\cite[Lemma~3]{WW19alog}, we provide a proof  in Appendix~\ref{sec:proof-direct-sum-ineq-sand-ren}.

\begin{lemma}[General direct-sum
properties]
\label{lem:direct-sum-ineq-sand-renyi} Let $\kappa\in\mathcal{S}%
(\mathcal{H}_{XA})$ and $\lambda\in\mathcal{L}_{+}(\mathcal{H}_{XA})$ be
classical--quantum, as defined as in Lemma~\ref{lem:direct-sum-rel-ent-prop}.
Then for $\alpha>1$%
\begin{equation}
\widetilde{D}_{\alpha}(\kappa\Vert\lambda)\geq D(p\Vert q)+\sum_{x}%
p(x)\widetilde{D}_{\alpha}(\kappa_{x}\Vert\lambda_{x}).
\label{eq:direct-sum-ineq-sand-renyi}%
\end{equation}
The same inequality holds for $\alpha>1$ with $\widetilde{D}_{\alpha}$ replaced by the Petz
or geometric R\'enyi relative entropy.
\end{lemma}


\subsection{Rains relative entropies}

We now turn to discussing the 
central quantity of this work.
Let $\rho\in\mathcal{S}(\mathcal{H}_{AB})$ denote a bipartite state acting on
a tensor-product Hilbert space $\mathcal{H}_{AB}:=\mathcal{H}_{A}%
\otimes\mathcal{H}_{B}$. 
To define it, we need to refer to the partial
transpose of a quantum state. To define it,
denote with $T(\cdot)\coloneqq\sum_{i,j}|i\rangle\!\langle j|(\cdot)|i\rangle
\!\langle j|$ 
the \emph{transpose map}, 
with $\{|i\rangle\}_{i}$ an orthonormal
basis for $\mathcal{H}_{B}$, so that $\operatorname{id}\otimes T:\mathcal{L}%
(\mathcal{H}_{AB})\rightarrow\mathcal{L}(\mathcal{H}_{AB})$ 
is the \emph{partial
transpose} map, given 
by $(\operatorname{id}\otimes T)(\cdot)\coloneqq\sum
_{i,j}(I\otimes|i\rangle\!\langle j|)(\cdot)(I\otimes|i\rangle\!\langle j|)$.
The Rains relative entropy has 
originally been defined as
\cite{R99,R01}
\begin{align}
R(\rho)  &  \coloneqq\inf_{\sigma\in\mathcal{S}(\mathcal{H}%
_{AB})}
\bigl(
D(\rho\Vert\sigma)
+ 
\log_2 \| 
(
\operatorname{id}\otimes T)
(\sigma)\|_1
\bigr),
\end{align}
where $\|\cdot \|_1$ denotes the trace norm. This expression has later been identified to be equal to the convex program
\cite{AdMVW02}
\begin{align}
R(\rho)  &  = \inf_{\sigma\in\operatorname{PPT}^{\prime}(\mathcal{H}%
_{AB})}D(\rho\Vert\sigma),
\label{eq:rains-rel-ent}\\
\operatorname{PPT}^{\prime}(\mathcal{H}_{AB})  &  \coloneqq\left\{
\sigma:\sigma\in\mathcal{L}_{+}(\mathcal{H}_{AB})\wedge\left\Vert \left(
\operatorname{id}\otimes T\right)  (\sigma)\right\Vert _{1}\leq1\right\}  .\nonumber
\end{align}
%
For
$\alpha\in(0,1)\cup(1,\infty)$,
we define with
\begin{equation}
\widetilde{R}_{\alpha}(\rho)\coloneqq\inf_{\sigma\in\operatorname{PPT}%
^{\prime}(\mathcal{H}_{AB})}\widetilde{D}_{\alpha}(\rho\Vert\sigma)
\label{eq:s-renyi-rains-rel-ent}%
\end{equation}
the sandwiched Rains relative entropy, 
satisfying \cite{TWW14}
\begin{equation}
\lim_{\alpha\rightarrow1}\widetilde{R}_{\alpha}(\rho)=R(\rho).
\label{eq:rains-a-limit-1}%
\end{equation}
By replacing $\widetilde{D}_{\alpha}$ in \eqref{eq:s-renyi-rains-rel-ent} with
either the Petz or geometric R\'{e}nyi relative entropy, we can construct
alternative Rains relative entropies   \cite{TWW14,KW20book}. 
They are subadditive in that
\begin{equation}
R(\rho_{0}\otimes\rho_{1})\leq R(\rho_{0})+R(\rho_{1}),\ \widetilde{R}%
_{\alpha}(\rho_{0}\otimes\rho_{1})\leq\widetilde{R}_{\alpha}(\rho
_{0})+\widetilde{R}_{\alpha}(\rho_{1}), \label{eq:Rains-subadditive}%
\end{equation}
for all
$\alpha\in(0,1)\cup(1,\infty)$,
where $\rho_{i}\in\mathcal{S}(\mathcal{H}_{A_{i}B_{i}})$ for $i\in\left\{
0,1\right\}  $ and the bipartite cut for $\rho_{0}\otimes\rho_{1}$ is $A_{0}A_{1}|B_{0}B_{1}$.
As a direct consequence of the joint convexity of quantum relative
entropy~\cite{LR73}, the Rains relative entropy is convex \cite{R01},
\begin{equation*}
R(\overline{\rho})\leq\sum_{x}p(x)R(\rho_{x}),
\end{equation*}
where $\overline{\rho}\coloneqq \sum_{x}p(x)\rho_{x}$, with $\{p(x)\}_{x}$ a
probability distribution and $\rho_{x}\in\mathcal{S}(\mathcal{H}_{AB})$ for
all $x$. Indeed, building on a statement from 
ref.~\cite[Lemma~3]{AdMVW02}, one can capture the 
Rains relative entropy in a more 
general way, proven in Appendix~\ref{sec:proof-rains-rewrite}.

\begin{lemma}[Scaling property]
\label{lem:rains-identity}Let $\mathbb{D}(\omega\Vert\tau)$ be a real-valued
function of $\omega\in\mathcal{S}(\mathcal{H})$ and $\tau\in\mathcal{L}%
_{+}(\mathcal{H})$, which 
satisfies 
\begin{equation}
\mathbb{D}(\omega\Vert c\tau)=\mathbb{D}(\omega\Vert\tau)-\log_{2}c
\label{eq:c-prop}%
\end{equation}
for all $c>0$. For $\rho\in\mathcal{S}(\mathcal{H}_{AB})$, define 
\begin{equation*}
\mathbb{R}(\rho) \coloneqq 
\inf_{\sigma\in\mathcal{S}(\mathcal{H}_{AB})}\left(
\mathbb{D}(\rho\Vert\sigma)+\log_{2}\!\left\Vert \left(  \operatorname{id}%
\otimes T\right)  \left(  \sigma\right)  \right\Vert _{1}\right)
\end{equation*}
as the
$\mathbb{D}$-Rains relative entropy.
Then
\begin{align}
\mathbb{R}(\rho)  &  =
\inf_{\sigma\in\operatorname{PPT}^{\prime
}(\mathcal{H}_{AB})}\mathbb{D}(\rho\Vert\sigma) \label{eq:rains-alt-exp}\\
&  =\inf_{\sigma\in\mathcal{L}_{+}(\mathcal{H}_{AB})}\left(  \mathbb{D}%
(\rho\Vert\sigma)+\log_{2}\!\left\Vert \left(  \operatorname{id}\otimes
T\right)  \left(  \sigma\right)  \right\Vert _{1}\right)  .
\label{eq:rains-alt-alt-exp}%
\end{align}

\end{lemma}


\begin{remark}[General Rains
relative entropies]
Since the property in \eqref{eq:c-prop} holds for the quantum relative entropy
and for the sandwiched R\'{e}nyi relative entropy for $\alpha\in
(0,1)\cup(1,\infty)$, the equalities in \eqref{eq:rains-alt-exp} and
\eqref{eq:rains-alt-alt-exp} hold for the Rains relative entropies constructed
from them (as given in \eqref{eq:rains-rel-ent} and
\eqref{eq:s-renyi-rains-rel-ent}). The same is true for the Rains relative
entropies constructed from the Petz and geometric R\'{e}nyi relative
entropies.
\end{remark}

\subsection{Completely PPT-preserving bipartite channels}

We now turn to describing meaningful classes of operations in the bipartite setting. From the outset, we generalize the physically meaningful LOCC,
but it should be clear that all that is said applies to this
practically important scenario.
Let $\mathcal{P}:\mathcal{L}(\mathcal{H}_{AB})\rightarrow\mathcal{L}%
(\mathcal{H}_{A^{\prime}B^{\prime}})$ denote a bipartite quantum channel, with
input systems $A$ and $B$ and output systems $A^{\prime}$ and $B^{\prime}$.
Physically, we think of one party Alice possessing the input system $A$ and
the output system $A^{\prime}$ and another party Bob in a distant laboratory
possessing the input system $B$ and the output system $B^{\prime}$. A
bipartite channel $\mathcal{P}$ is \emph{completely PPT\ preserving} \cite{R99,CVGG17} 
if the map $\left(  \operatorname{id}\otimes T\right)
\circ\mathcal{P}\circ\left(  \operatorname{id}\otimes T\right)  $ is
completely positive.
It is well known that
\begin{equation}
\sigma\in\operatorname{PPT}^{\prime}(\mathcal{H}_{AB}) \quad\Rightarrow
\quad\mathcal{P}(\sigma)\in\operatorname{PPT}^{\prime}(\mathcal{H}_{A^{\prime
}B^{\prime}}), \label{eq:PPT'-closure}%
\end{equation}
and this implies from \eqref{eq:DP}
that%
\begin{equation}
R(\rho)\geq R(\mathcal{P}(\rho)),\qquad\widetilde{R}_{\alpha}(\rho
)\geq\widetilde{R}_{\alpha}(\mathcal{P}(\rho)), \label{eq:PPT-monotone}%
\end{equation}
where the latter holds for all $\alpha\in[1/2,1)\cup(1,\infty)$. Indeed, to
see \eqref{eq:PPT'-closure}, let $\sigma\in\operatorname{PPT}^{\prime
}(\mathcal{H}_{AB})$ and consider that
\begin{align}
&  \left\Vert \left(  \operatorname{id}\otimes T\right)  (\mathcal{P}%
(\sigma))\right\Vert _{1}=\left\Vert (\left(  \operatorname{id}\otimes T\right)  \circ
\mathcal{P}\circ\left(  \operatorname{id}\otimes T\right)  )(\left(
\operatorname{id}\otimes T\right)  (\sigma))\right\Vert _{1}\nonumber\\
&  \leq\left\Vert \left(  \operatorname{id}\otimes T\right)  (\sigma
)\right\Vert _{1}\leq1,
\end{align}
where the first inequality follows because the trace norm does not increase
under the action of a completely positive trace-preserving map (in this case
$\left(  \operatorname{id}\otimes T\right)  \circ\mathcal{P}\circ\left(
\operatorname{id}\otimes T\right)  $).


The significance of the set of completely
PPT-preserving channels \cite{R99,R01,APE03}
is that it contains the operationally relevant set of
local operations and classical communication (LOCC) channels
\cite{BDSW96,CLM+14}. The constraints specifying the former set of  channels are semi-definite, and it is thus much
easier in a computational sense to optimize an objective function over this
set of channels than to optimize over the set of LOCC channels (see, e.g., \cite{SW21}).
Generalizing completely PPT-preserving 
channels, 
physically well motivated, and of relevance for us here,
is the set of \emph{selective PPT operations}.

\begin{definition}
[Selective PPT operation]\label{def:selective-PPT} The set $\{\mathcal{P}%
_{x}\}_{x}$ constitutes a selective PPT operation if $\mathcal{P}_{x}
\in\operatorname{CP}(\mathcal{H}_{AB},\mathcal{H}_{A^{\prime}B^{\prime}})$ for
all $x$, $\left(  \operatorname{id}\otimes T\right)  \circ\mathcal{P}_{x}%
\circ\left(  \operatorname{id}\otimes T\right)  \in\operatorname{CP}%
(\mathcal{H}_{AB},\mathcal{H}_{A^{\prime}B^{\prime}})$ for all~$x$, and the
sum map $\sum_{x}\mathcal{P}_{x}$ is trace preserving.
\end{definition}

When performed on a quantum state $\rho\in\mathcal{S}(\mathcal{H}_{AB})$, the
selective PPT operation $\{\mathcal{P}_{x}\}_{x}$ can be understood as a
particular kind of quantum instrument, which outputs the state $\mathcal{P}%
_{x}(\rho)/p(x)$ with probability $p(x)=\operatorname{Tr}[\mathcal{P}_{x}%
(\rho)]$. Similar to the channel case discussed above, the set of selective
PPT operations contains the set of \emph{selective LOCC operations}
  \cite{CLM+14}.
%
Let $\sigma\in\mathcal{S}(\mathcal{H}_{AB})$ be a PPT\ state, i.e., a state
satisfying $(\operatorname{id}\otimes T)(\sigma)\geq0$. Let $\sigma
_{x}\coloneqq\mathcal{P}_{x}(\sigma)/q(x)$ with
$q(x)\coloneqq\operatorname{Tr}[\mathcal{P}_{x}(\sigma)]$. It follows that
$\sigma_{x}$ is also a PPT\ state because%
\begin{multline}
\left(  \operatorname{id}\otimes T\right)  (\sigma_{x})=\frac{1}{q(x)}\left(
\operatorname{id}\otimes T\right)  \mathcal{P}_{x}(\sigma)\\
=\frac{1}{q(x)}\left(  \left(  \operatorname{id}\otimes T\right)
\circ\mathcal{P}_{x}\left(  \operatorname{id}\otimes T\right)  \right)
\left(  \left(  \operatorname{id}\otimes T\right)  (\sigma)\right)  \geq0,
\end{multline}
where the inequality follows because $\left(  \operatorname{id}\otimes
T\right)  \circ\mathcal{P}_{x}\left(  \operatorname{id}\otimes T\right)  $ is
completely positive by definition and $\left(  \operatorname{id}\otimes
T\right)  (\sigma)\geq0$. Now suppose that $\sigma\in\operatorname{PPT}%
^{\prime}(\mathcal{H}_{AB})$. With $\sigma_{x}$ and $q(x)$ as defined above,
we know that the following inequality holds%
\begin{equation}
\left\Vert \left(  \operatorname{id}\otimes T\right)  (\sigma)\right\Vert
_{1}\geq\sum_{x}q(x)\left\Vert \left(  \operatorname{id}\otimes T\right)
(\sigma_{x})\right\Vert _{1}, \label{eq:negativity-monotone}%
\end{equation}
as a consequence of 
ref.~\cite[Eq.~(8)]{Plenio2005b}, 
following the approach given in 
ref.~\cite[Proposition~2.1]{E01} (see also
ref.~\cite[Proposition~9.10]{KW20book}). However, it is not clear if $\sigma_{x}%
\in\operatorname{PPT}^{\prime}(\mathcal{H}_{AB})$ for all~$x$. This is the
main obstacle to overcome in proving Theorem~\ref{thm:main} of the next
section, and we do so by exploiting Lemma~\ref{lem:rains-identity} and
properties of generalized relative entropies.

\section{Selective PPT\ monotonicity}

We now derive one of the main results of this work: All the Rains relative
entropies discussed thus far are selective PPT monotones. By the fact that
selective LOCC operations are contained in the set of selective PPT
operations, it follows that these quantities are also selective LOCC
monotones. Again, as such, 
they are quantities meaningfully
quantifying entanglement, according to the stronger selective notion, in a 
computable and conservative way.
This property is also one of the main ones needed to establish the Rains
relative entropies as upper bounds on probabilistic approximate distillable
entanglement, which we do in Section~\ref{sec:PADE}.
We provide a 
general statement in Theorem~\ref{thm:main} below, which
applies to the standard, 
as well as other, Rains relative entropies.

\begin{theorem}[Selective
entanglement monotonicity]
\label{thm:main}Let $\boldsymbol{D}(\omega\Vert\tau)$ be a real-valued
function of $\omega\in\mathcal{S}(\mathcal{H})$ and $\tau\in\mathcal{L}%
_{+}(\mathcal{H})$ for which  the data-processing inequality holds:
\begin{equation*}
\boldsymbol{D}(\omega\Vert\tau)\geq\boldsymbol{D}(\mathcal{N}(\omega
)\Vert\mathcal{N}(\tau))
\end{equation*}
for $\mathcal{N\in}\operatorname{CPTP}(\mathcal{H},\mathcal{H}^{\prime})$, and
for which the following holds%
\begin{equation}
\boldsymbol{D}(\kappa\Vert\lambda)\geq D(p\Vert q)+\sum_{x}p(x)\boldsymbol{D}%
(\kappa_{x}\Vert\lambda_{x}), \label{eq:direct-sum-ineq-GD}%
\end{equation}
for classical--quantum $\kappa\in\mathcal{S}(\mathcal{H}_{XA})$ and
$\lambda\in\mathcal{L}_{+}(\mathcal{H}_{XA})$, as defined in
Lemma~\ref{lem:direct-sum-rel-ent-prop}. For $\rho\in\mathcal{S}%
(\mathcal{H}_{AB})$, define %
\begin{equation}
\boldsymbol{R}(\rho)\coloneqq\inf_{\sigma\in\mathcal{S}(\mathcal{H}_{AB})}
\left(\boldsymbol{D}(\rho\Vert\sigma) + \log_{2}\!\left\Vert \left(
\operatorname{id}\otimes T\right)  (\sigma)\right\Vert _{1}\right)
\label{eq:bold-d-rains}
\end{equation}
as
the $\boldsymbol{D}$-Rains relative entropy.
Then the $\boldsymbol{D}$-Rains relative entropy $\boldsymbol{R}(\rho)$ is a
selective PPT monotone; i.e., it satisfies%
\begin{equation*}
\boldsymbol{R}(\rho)\geq\sum_{x:p(x)>0}p(x)\boldsymbol{R}(\rho_{x}),
\end{equation*}
where%
\begin{equation*}
p(x)\coloneqq\operatorname{Tr}[\mathcal{P}_{x}(\rho)],\qquad\rho
_{x}\coloneqq\frac{\mathcal{P}_{x}(\rho)}{p(x)},
\end{equation*}
and $\{\mathcal{P}_{x}\}_{x}$ is a selective PPT operation (see
Definition~\ref{def:selective-PPT}).
\end{theorem}

\begin{proof}
Consider the quantum channel%
\begin{equation*}
\mathcal{P}(\cdot)\coloneqq\sum_{x}|x\rangle\!\langle x|\otimes\mathcal{P}%
_{x}(\cdot).
\end{equation*}
Let $\sigma\in\mathcal{S}(\mathcal{H}_{AB})$, and let us define the
probability distribution $\{q(x)\}_{x}$ and set $\{\sigma_{x}\}_{x}$ 
with $\sigma_{x}\in\mathcal{S}(\mathcal{H}_{A^{\prime}B^{\prime}})$
as%
\begin{equation*}
q(x)\coloneqq\operatorname{Tr}[\mathcal{P}_{x}(\sigma)],\qquad\sigma
_{x}\coloneqq\frac{1}{q(x)}\mathcal{P}_{x}(\sigma).
\end{equation*}
From the
data-processing inequality for $\boldsymbol{D}$ and
\eqref{eq:direct-sum-ineq-GD}, we find 
\begin{align}
\boldsymbol{D}(\rho\Vert\sigma)  &  \geq\boldsymbol{D}(\mathcal{P}(\rho
)\Vert\mathcal{P}(\sigma))\nonumber\\
&  =\boldsymbol{D}\!\left(  \sum_{x}|x\rangle\!\langle x|\otimes
\mathcal{P}_{x}(\rho)\middle\Vert\sum_{x}|x\rangle\!\langle x|\otimes
\mathcal{P}_{x}(\sigma)\right) \nonumber\\
&  =\boldsymbol{D}\!\left(  \sum_{x}p(x)|x\rangle\!\langle x|\otimes\rho
_{x}\middle\Vert\sum_{x}q(x)|x\rangle\!\langle x|\otimes\sigma_{x}\right)
\nonumber\\
&  \geq D(p\Vert q)+\sum_{x}p(x)\boldsymbol{D}(\rho_{x}\Vert\sigma_{x}).
\end{align}
Putting this together with \eqref{eq:negativity-monotone} and defining
$\mathcal{X}^{+}\coloneqq\left\{  x:p(x)>0\right\}  $, we find that%
\begin{align}
&  \boldsymbol{D}(\rho\Vert\sigma)+\log_{2}\!\left\Vert \left(
\operatorname{id}\otimes T\right)  (\sigma)\right\Vert _{1}
\geq\sum_{x\in\mathcal{X}^{+}}p(x)\boldsymbol{D}(\rho_{x}\Vert\sigma
_{x})\nonumber\\
&  \qquad +D(p\Vert q) +\log_{2}\!\left(  \sum_{x}q(x)\left\Vert \left(  \operatorname{id}%
\otimes T\right)  (\sigma_{x})\right\Vert _{1}\right) \\
&  \geq\sum_{x\in\mathcal{X}^{+}}p(x)\boldsymbol{D}(\rho_{x}\Vert\sigma
_{x})+\sum_{x\in\mathcal{X}^{+}}p(x)\log_{2}\!\left\Vert \left(
\operatorname{id}\otimes T\right)  (\sigma_{x})\right\Vert _{1}\nonumber\\
&  
\geq\sum_{x\in\mathcal{X}^{+}}p(x)\boldsymbol{R}(\rho_{x}).\nonumber
\end{align}
The inequality follows from the definition in \eqref{eq:bold-d-rains}. The
second inequality follows because%
\begin{align}
&  D(p\Vert q)+\log_{2}\!\left(  \sum_{x}q(x)\left\Vert \left(
\operatorname{id}\otimes T\right)  (\sigma_{x})\right\Vert _{1}\right)
\nonumber\\
&  \geq D(p\Vert q)+\log_{2}\!\left(  \sum_{x\in\mathcal{X}^{+}}%
p(x)\frac{q(x)}{p(x)}\left\Vert \left(  \operatorname{id}\otimes T\right)
(\sigma_{x})\right\Vert _{1}\right) \nonumber\\
&  \geq D(p\Vert q)+\sum_{x\in\mathcal{X}^{+}}p(x)\log_{2}\!\left(
\frac{q(x)}{p(x)}\left\Vert \left(  \operatorname{id}\otimes T\right)
(\sigma_{x})\right\Vert _{1}\right) \nonumber\\
&  =\sum_{x\in\mathcal{X}^{+}}p(x)
\left(
\log_{2}\!\left(  \frac{p(x)}{q(x)}\right)
+
\log_{2}\!\left(  \frac{q(x)}%
{p(x)}\left\Vert \left(  \operatorname{id}\otimes T\right)  (\sigma
_{x})\right\Vert _{1}\right)\right) \nonumber\\
&  =\sum_{x\in\mathcal{X}^{+}}p(x)\log_{2}\!\left\Vert \left(
\operatorname{id}\otimes T\right)  (\sigma_{x})\right\Vert _{1},
\end{align}
where we have used the concavity of the logarithm. Since 
\begin{equation}
\boldsymbol{D}(\rho\Vert\sigma)+\log_{2}\!\left\Vert \left(  \operatorname{id}%
\otimes T\right)  (\sigma)\right\Vert _{1}\geq\sum_{x\in\mathcal{X}^{+}%
}p(x)\boldsymbol{R}(\rho_{x}) \label{eq:final-step-monotone-proof}%
\end{equation}
holds for arbitrary $\sigma\in\mathcal{S}(\mathcal{H}_{AB})$, we conclude 
by taking an infimum over every $\sigma\in\mathcal{S}%
(\mathcal{H}_{AB})$ and applying Lemma~\ref{lem:rains-identity}.
\end{proof}

\begin{corollary}[Rains relative entropies as selective
entanglement monotones]
\label{rem:selective-PPT-actual-rel-ents}The Rains relative entropies
constructed from the quantum relative entropy, and the Petz (for
$\alpha\in(1,2]$), geometric
(for $\alpha\in(1,2]$), and
sandwiched (for $\alpha>1$)
R\'{e}nyi relative entropies are selective PPT monotones. 
\end{corollary}


\begin{proof}
All quantities satisfy the hypotheses of
Theorem~\ref{thm:main}.
\end{proof}


The proof of the following corollary, 
presented in Appendix~\ref{sec:proof-inv-class-comm}, justifies a claim made in ref.~\cite[Lemma~24]%
{KW17a}.

\begin{corollary}[Flags property]
\label{cor:inv-class-comm} The Rains relative entropy  possesses the \textquotedblleft
flags\textquotedblright\ property
\cite{HHHH09}.
That is, 
for a probability distribution
$\{p(x)\}_{x}$ and 
$\omega_{x}%
\in\mathcal{S}(\mathcal{H}_{AB})$,
let $\omega\in\mathcal{S}(\mathcal{H}_{XAB})$ be 
of the form%
\begin{equation*}
\omega=\sum_{x}p(x)|x\rangle\!\langle x|\otimes\omega_{x}.
\end{equation*}
Then,
for either of the bipartite cuts 
 $XA|B$ or $A|XB$,
\begin{equation*}
R(\omega)=\sum_{x}p(x)R(\omega_{x}).
\end{equation*}
\end{corollary}


\section{Probabilistic Approximate Distillation}

\label{sec:PADE}

In this final section, 
we turn our attention to an important application of Theorem~\ref{thm:main}. We prove that the Rains
relative entropies 
are upper bounds on the non-asymptotic
and asymptotic 
probabilistic approximate distillable entanglement (PADE). This
result strengthens the original result of ref.~\cite{R01} because the PADE cannot be smaller than the standard distillable entanglement
 \cite{HHHH09,KW20book}. 
The notion of PADE and 
such concepts 
in other resource theories has become increasingly relevant
\cite{RSTEDW18,FWLRA18,FL20,Regula2021,R21,KDS21}, due to the connection with
experiment and the fact that desired resource states are in practice generated
only approximately and with a certain probability heralded by a classical
signal.

\subsection{Non-asymptotic case}

Given a bipartite state $\rho\in\mathcal{S}(\mathcal{H}_{AB})$ and error
parameter $\varepsilon\in\left[  0,1\right]  $, we define the non-asymptotic
\emph{probabilistic approximate distillable entanglement (PADE)}
as%
\begin{multline}
E_{d}^{\varepsilon}(\rho)\coloneqq\\
\sup_{\substack{p\in\left[  0,1\right]  ,\mathcal{L},\\
d \in \mathbb{N}}}\left\{
\begin{array}
[c]{c}%
p\log_{2}d:\\
\mathcal{L}(\rho)=p|1\rangle\!\langle1|\otimes\widetilde{\Phi}^{d}+\left(
1-p\right)  |0\rangle\!\langle0|\otimes\sigma,\\
F(\widetilde{\Phi}^{d},\Phi^{d})\geq1-\varepsilon,\ \sigma\in\mathcal{S}(\mathcal{H}_{\hat{A}\hat{B}})
\end{array}
\right\}  ,\nonumber
\end{multline}
where the optimization is over $\mathcal{L}\in\operatorname{LOCC}%
(\mathcal{H}_{AB},\mathcal{H}_{X\hat{A}\hat{B}})$, the probability $p\in[0,1]$, and the dimension $d=\dim(\mathcal{H}_{\hat{A}})=\dim(\mathcal{H}_{\hat{B}})$. Also, $F(\omega,\tau
)\coloneqq\left\Vert \sqrt{\omega}\sqrt{\tau}\right\Vert _{1}^{2}$ is the
fidelity of states $\omega,\tau\in\mathcal{S}(\mathcal{H})$ \cite{U76}, and
$\Phi^{d}\coloneqq|\Phi^{d}\rangle\!\langle\Phi^{d}|$
is the maximally entangled state, 
\begin{equation}
\qquad|\Phi^{d}%
\rangle\coloneqq d^{-1/2}\sum_{i}|i\rangle|i\rangle,
\label{eq:max-ent-def}%
\end{equation}
with $\{|i\rangle\}_{i}$ an orthonormal basis. The interpretation of
$E_{d}^{\varepsilon}(\rho)$ is that it is equal to the expected number of
$\varepsilon$-approximate e-bits that can be generated from the state $\rho$
by means of LOCC, where an e-bit is $\Phi^{2}$. The standard non-asymptotic
distillable entanglement of $\rho$ is defined in the same way as above,
however, with the exception that there is no optimization over $p$, and it is
instead simply set to $p=1$ \cite{BD10a}.

\begin{theorem}[Non-asymptotic upper bound]
\label{thm:non-asymp-ED-upper-bnd}For $\varepsilon\in\lbrack0,1)$ and $\rho
\in\mathcal{S}(\mathcal{H}_{AB})$, the following upper bound holds for all
$\alpha>1$:%
\begin{equation}
E_{d}^{\varepsilon}(\rho)\leq\widetilde{R}_{\alpha}(\rho)+\frac{\alpha}%
{\alpha-1}\log_{2}\!\left(  \frac{1}{1-\varepsilon}\right)  .
\label{eq:ed-up-bnd}%
\end{equation}

\end{theorem}

\begin{proof}
Let $\mathcal{L}\in\operatorname{LOCC}(\mathcal{H}_{AB},\mathcal{H}_{X\hat
{A}\hat{B}})$ and $p\in\left[  0,1\right]  $ be the elements of an arbitrary
non-asymptotic PADE\ protocol. For $\alpha>1$, we find that%
\begin{align}
\widetilde{R}_{\alpha}(\rho)  &  \geq p\widetilde{R}_{\alpha}(\widetilde{\Phi
}^{d})+\left(  1-p\right)  \widetilde{R}_{\alpha}(\sigma)\nonumber
 \geq p\widetilde{R}_{\alpha}(\widetilde{\Phi}^{d})\nonumber\\
&  \geq p\left[  \frac{\alpha}{\alpha-1}\log_{2}(1-\varepsilon)+\log
_{2}d\right] \nonumber\\
&  \geq\frac{\alpha}{\alpha-1}\log_{2}(1-\varepsilon)+p\log_{2}d,
\end{align}
where the first inequality follows from applying Theorem~\ref{thm:main} and
Corollary~\ref{rem:selective-PPT-actual-rel-ents}\ to $\mathcal{L}$ and $\rho
$. The second inequality follows because $\widetilde{R}_{\alpha}(\sigma)\geq
0$, and the third inequality follows from applying
Lemma~\ref{lem:rains-ineq-approx-max-ent} in
Appendix~\ref{sec:rains-up-bnd-approx-ebits}. Rewriting the inequality
$\widetilde{R}_{\alpha}(\rho)\geq\frac{\alpha}{\alpha-1}\log_{2}%
(1-\varepsilon)+p\log_{2}d$, we find, for all $\alpha>1$,
\begin{equation*}
p\log_{2}d\leq\widetilde{R}_{\alpha}(\rho)+\frac{\alpha}{\alpha-1}\log
_{2}\!\left(  \frac{1}{1-\varepsilon}\right)  .
\end{equation*}
Since the upper bound depends only on $\rho$, $\varepsilon$, and $\alpha$ and
it holds for all $\mathcal{L}\in\operatorname{LOCC}(\mathcal{H}_{AB}%
,\mathcal{H}_{X\hat{A}\hat{B}})$ and $p\in\left[  0,1\right]  $, we conclude
\eqref{eq:ed-up-bnd}\ after applying the definition of $E_{d}^{\varepsilon
}(\rho)$.
\end{proof}

\subsection{Asymptotic case}

We define the asymptotic PADE as 
\begin{equation}
E_{d}(\rho)\coloneqq\inf_{\varepsilon\in\left(  0,1\right)  }\liminf
_{n\rightarrow\infty}\frac{1}{n}E_{d}^{\varepsilon}(\rho^{\otimes n}),
\label{eq:asymp-PADE}%
\end{equation}
as well as the strong converse PADE as%
\begin{equation}
\widetilde{E}_{d}(\rho)\coloneqq\sup_{\varepsilon\in\left(  0,1\right)
}\limsup_{n\rightarrow\infty}\frac{1}{n}E_{d}^{\varepsilon}(\rho^{\otimes n}).
\label{eq:strong-conv-PADE}%
\end{equation}
The inequality $E_{d}(\rho)\leq\widetilde{E}_{d}(\rho)$ is obvious from the definitions.

\begin{theorem}[Strong converse]
\label{thm:strong-conv-Rains} For $\rho\in\mathcal{S}(\mathcal{H}_{AB})$, one finds 
$
\widetilde{E}_{d}(\rho)\leq R(\rho)$.

\end{theorem}

\begin{proof}
By applying Theorem~\ref{thm:non-asymp-ED-upper-bnd}, we find that the
following upper bound holds for all $\varepsilon\in\lbrack0,1)$ and $\alpha
>1$.%
\begin{align}
\frac{1}{n}E_{d}^{\varepsilon}(\rho^{\otimes n})  &  \leq\frac{1}{n}%
\widetilde{R}_{\alpha}(\rho^{\otimes n})+\frac{\alpha}{n(\alpha-1)}\log
_{2}\!\left(  \frac{1}{1-\varepsilon}\right) \\
&  \leq\widetilde{R}_{\alpha}(\rho)+\frac{\alpha}{n(\alpha-1)}\log
_{2}\!\left(  \frac{1}{1-\varepsilon}\right)  ,
\end{align}
where we have used subadditivity of the sandwiched Rains relative entropy 
\eqref{eq:Rains-subadditive}. Taking the limit as $n\rightarrow\infty$, we find that the following holds for all
$\varepsilon\in\lbrack0,1)$ and $\alpha>1$
$
\limsup_{n\rightarrow\infty}\frac{1}{n}E_{d}^{\varepsilon}(\rho^{\otimes
n})\leq\widetilde{R}_{\alpha}(\rho).
$
The limit
$\alpha\rightarrow1$ yields
\begin{equation}
\limsup_{n\rightarrow\infty}\frac{1}{n}E_{d}^{\varepsilon}(\rho^{\otimes
n})\leq R(\rho), \label{eq:key-ineq-str-conv-rains}%
\end{equation}
where we have used \eqref{eq:rains-a-limit-1}. Then the desired inequality follows
because the above bound holds for all $\varepsilon\in\left(  0,1\right)  $.
\end{proof}

We can regularize \cite{PhysRevLett.87.217902} the bound above to make it even tighter, with the proof being presented in  Appendix~\ref{sec:reg-Rains-up-bnd}.
In Appendix~\ref{sec:weak-conv-bnd}, we provide a non-asymptotic bound
which leads
to a weak-converse upper bound on the PADE. 

\begin{corollary}[Regularized Rains bound]
\label{cor:reg-rains-bnd} For $\rho\in\mathcal{S}(\mathcal{H}_{AB})$, 
\begin{equation*}
\widetilde{E}_{d}(\rho)\leq\inf_{\ell\in\mathbb{N}}\frac{1}{\ell}%
R(\rho^{\otimes\ell}).
\end{equation*}

\end{corollary}


\begin{remark}[Tighter bounds]
Due to the subadditivity of the Rains relative entropy  
\eqref{eq:Rains-subadditive}, the sequence $\{\frac{1}{\ell}R(\rho
^{\otimes\ell})\}_{\ell\in\mathbb{N}}$ is monotone decreasing, so that one
obtains upper bounds that become tighter with increasing $\ell$.
\end{remark}


In summary, we have proven that various Rains relative entropies are selective PPT entanglement monotones and have established this quantity as
a meaningful conservative entanglement measure. We have also shown that they give upper bounds on the probabilistic approximate distillable entanglement, in both the non-asymptotic and asymptotic settings.

{\footnotesize
Acknowledgements: This work has been initiated by
J.~E.~answering a technical question put forth
by M.~M.~W.~on Twitter. We thank Kr{\"u}mel for proofreading.
J.~E.~acknowledges support from the 
BMBF (QR.X) and the Einstein Research Unit on Quantum Devices. M.~M.~W.~acknowledges support from 
NSF Grant No.~1907615.}

\bibliographystyle{IEEEtran}
\bibliography{RefNew}

\newpage

\appendices

\section{Proof of Lemma~\ref{lem:direct-sum-ineq-sand-renyi}}

\label{sec:proof-direct-sum-ineq-sand-ren}

The inequality in \eqref{eq:direct-sum-ineq-sand-renyi} can be understood as a
rewriting of  ref.~\cite[Lemma~3]{WW19alog}. We give a full proof here for
completeness. Consider that%
\begin{align}
&  \widetilde{D}_{\alpha}(\kappa\Vert\lambda)\nonumber\\
&  =\frac{1}{\alpha-1}\log_{2}\!\left(  \sum_{x}p^{\alpha}(x)q^{1-\alpha
}(x)\widetilde{Q}_{\alpha}(\kappa_{x}\Vert\lambda_{x})\right) \nonumber\\
&  =\frac{1}{\alpha-1}\log_{2}\!\left(  \sum_{x}p(x)p^{\alpha-1}%
(x)q^{1-\alpha}(x)\widetilde{Q}_{\alpha}(\kappa_{x}\Vert\lambda_{x})\right)
\nonumber\\
&  =\frac{1}{\alpha-1}\log_{2}\!\left(  \sum_{x}p(x)\left(  \frac{p(x)}%
{q(x)}\right)  ^{\alpha-1}\widetilde{Q}_{\alpha}(\kappa_{x}\Vert\lambda
_{x})\right) \nonumber\\
&  \geq\sum_{x}p(x)\left[  \frac{1}{\alpha-1}\log_{2}\!\left(  \left(
\frac{p(x)}{q(x)}\right)  ^{\alpha-1}\widetilde{Q}_{\alpha}(\kappa_{x}%
\Vert\lambda_{x})\right)  \right] \nonumber\\
&  =\sum_{x}p(x)\left[  \log_{2}\!\left(  \frac{p(x)}{q(x)}\right)  +\frac
{1}{\alpha-1}\log_{2}\widetilde{Q}_{\alpha}(\kappa_{x}\Vert\lambda_{x})\right]
\nonumber\\
&  =D(p\Vert q)+\sum_{x}p(x)\widetilde{D}_{\alpha}(\kappa_{x}\Vert\lambda
_{x}).
\end{align}
The sole inequality above follows from concavity of the logarithm. The proof
of the same inequality for the Petz and geometric Renyi relative entropies
for $\alpha>1$ follows the same approach.

\section{Proof of Lemma~\ref{lem:rains-identity}}

\label{sec:proof-rains-rewrite}

Our goal is to prove the equality
\begin{multline}
\inf_{\sigma\in\operatorname{PPT}^{\prime
}(\mathcal{H}_{AB})}\mathbb{D}(\rho\Vert\sigma) 
\\
=
    \inf_{\sigma\in\mathcal{S}(\mathcal{H}_{AB})}\left(
\mathbb{D}(\rho\Vert\sigma)+\log_{2}\!\left\Vert \left(  \operatorname{id}%
\otimes T\right)  \left(  \sigma\right)  \right\Vert _{1}\right)
\label{eq:Rains-identity-app}
\end{multline}
To see the inequality $\leq$ in \eqref{eq:Rains-identity-app}, let $\sigma
\in\mathcal{S}(\mathcal{H}_{AB})$, and consider that
\begin{align}
 \inf_{\sigma\in\operatorname{PPT}^{\prime
}(\mathcal{H}_{AB})}\mathbb{D}(\rho\Vert\sigma)    &  \leq\mathbb{D}(\rho\Vert\sigma/\left\Vert \left(
\operatorname{id}\otimes T\right)  (\sigma)\right\Vert _{1})\nonumber\\
&  =\mathbb{D}(\rho\Vert\sigma)+\log_{2}\!\left\Vert \left(  \operatorname{id}%
\otimes T\right)  (\sigma)\right\Vert _{1},
\end{align}
which follows because $\sigma/\left\Vert \left(  \operatorname{id}\otimes
T\right)  (\sigma)\right\Vert _{1}\in\operatorname{PPT}^{\prime}%
(\mathcal{H}_{AB})$ and by applying \eqref{eq:c-prop}. Since the inequality
holds for all $\sigma\in\mathcal{S}(\mathcal{H}_{AB})$, we conclude that
\begin{multline}
\inf_{\sigma\in\operatorname{PPT}^{\prime
}(\mathcal{H}_{AB})}\mathbb{D}(\rho\Vert\sigma) 
\leq\\
\inf_{\sigma\in\mathcal{S}(\mathcal{H}_{AB})}\left[
\mathbb{D}(\rho\Vert\sigma)+\log_{2}\!\left\Vert \left(  \operatorname{id}%
\otimes T\right)  (\sigma)\right\Vert _{1}\right]  . \label{eq:one-way-for-id}%
\end{multline}
For the other inequality $\geq$ in \eqref{eq:rains-alt-exp}, let $\sigma
\in\operatorname{PPT}^{\prime}(\mathcal{H}_{AB})$. Then consider that%
\begin{align}
\mathbb{D}(\rho\Vert\sigma)  &  =\mathbb{D}(\rho\Vert\sigma/\operatorname{Tr}%
[\sigma])-\log_{2}\operatorname{Tr}[\sigma]\nonumber\\
&  =\mathbb{D}(\rho\Vert\sigma/\operatorname{Tr}[\sigma])-\log_{2}%
\operatorname{Tr}[\sigma]\nonumber\\
&  \qquad+\log_{2}\!\left\Vert \left(  \operatorname{id}\otimes T\right)
(\sigma)\right\Vert _{1}-\log_{2}\!\left\Vert \left(  \operatorname{id}\otimes
T\right)  (\sigma)\right\Vert _{1}\nonumber\\
&  =\mathbb{D}(\rho\Vert\sigma/\operatorname{Tr}[\sigma])+\log_{2}\!\left\Vert
\left(  \operatorname{id}\otimes T\right)  \!\left(  \frac{\sigma
}{\operatorname{Tr}[\sigma]}\right)  \right\Vert _{1}\nonumber\\
&  \qquad-\log_{2}\!\left\Vert \left(  \operatorname{id}\otimes T\right)
(\sigma)\right\Vert _{1}\nonumber\\
&  \geq\mathbb{D}(\rho\Vert\sigma/\operatorname{Tr}[\sigma])+\log
_{2}\left\Vert \left(  \operatorname{id}\otimes T\right)  \!\left(
\frac{\sigma}{\operatorname{Tr}[\sigma]}\right)  \right\Vert _{1}\nonumber\\
&  \geq\inf_{\omega\in\mathcal{S}(\mathcal{H}_{AB})}\left[  \mathbb{D}%
(\rho\Vert\omega)+\log_{2}\!\left\Vert \left(  \operatorname{id}\otimes
T\right)  (\omega)\right\Vert _{1}\right]  .
\end{align}
The first equality follows from \eqref{eq:c-prop} and the first inequality
from the fact that $\left\Vert \left(  \operatorname{id}\otimes T\right)
(\sigma)\right\Vert _{1}\leq1$. Since the inequality holds for arbitrary
$\sigma\in\operatorname{PPT}^{\prime}(\mathcal{H}_{AB})$, we conclude that%
\begin{multline}
\inf_{\sigma\in\operatorname{PPT}^{\prime}(\mathcal{H}_{AB}%
)}\mathbb{D}(\rho\Vert\sigma)\label{eq:other-way-for-id}\\
\geq\inf_{\omega\in\mathcal{S}(\mathcal{H}_{AB})}\left[  \mathbb{D}(\rho
\Vert\omega)+\log_{2}\!\left\Vert \left(  \operatorname{id}\otimes T\right)
(\omega)\right\Vert _{1}\right]  .
\end{multline}
Putting together \eqref{eq:one-way-for-id} and \eqref{eq:other-way-for-id}, we
conclude the desired equality in \eqref{eq:Rains-identity-app}. The proof of the equality expressed in 
\eqref{eq:rains-alt-alt-exp} is the same.

\section{Proof of Corollary~\ref{cor:inv-class-comm}}

\label{sec:proof-inv-class-comm}

Convexity of $R$ and its monotonicity under the local channel $(\cdot
)\mapsto|x\rangle\!\langle x|\otimes(\cdot)$ (applying
\eqref{eq:PPT-monotone}) imply that%
\begin{equation*}
R(\omega)\leq\sum_{x}p(x)R(|x\rangle\!\langle x|\otimes\omega_{x})\leq\sum
_{x}p(x)R(\omega_{x}).
\end{equation*}
Defining the local projection $\overline{\Delta}_{x}:\mathcal{L}%
_{+}(\mathcal{H}_{X})\rightarrow\mathcal{L}_{+}
(\mathcal{H}_{X})$ as%
\begin{equation*}
(\cdot)\mapsto|x\rangle\!\langle x|(\cdot)|x\rangle\!\langle x|,
\end{equation*}
and observing that the sum map $\sum_{x}\overline{\Delta}_{x}$ is trace
preserving implies that the set $\{\overline{\Delta}_{x}\otimes
\operatorname{id}:\mathcal{L}_{+}(\mathcal{H}_{XAB})\rightarrow\mathcal{L}%
_{+}(\mathcal{H}_{XAB})\}_{x}$ constitutes a selective PPT\ operation.
Applying Theorem~\ref{thm:main}, we find that%
\begin{equation*}
R(\omega)\geq\sum_{x}p(x)R(|x\rangle\!\langle x|\otimes\omega_{x})\geq\sum
_{x}p(x)R(\omega_{x}),
\end{equation*}
where the last inequality follows because the partial trace over
$\mathcal{H}_{X}$ is a local channel and $R$ is monotone under the action of
local channels.

\section{Proof of Inequality for Rains Relative Entropy of Approximate E-bits}

\label{sec:rains-up-bnd-approx-ebits}

In this appendix, we prove Lemma~\ref{lem:rains-ineq-approx-max-ent}, which
states that the sandwiched Rains relative entropy and a correction term
provide an upper bound on the number of approximate e-bits that are available
in a state $\widetilde{\Phi}^{d}$ close in fidelity to an ideal maximally
entangled state $\Phi^{d}$. We provide two different proofs of this lemma:\ a
first that is similar to ref.~\cite[Proposition~4]{TWW14} and a second in terms of a pseudo-continuity bound 
for the sandwiched Rains relative entropy, given in Lemma~\ref{lem:pseudo-cont-sandwiched-renyi}\ below.

\begin{lemma}[Approximate
e-bits]
\label{lem:rains-ineq-approx-max-ent}Let $\varepsilon\in\left[  0,1\right]  $,
$\widetilde{\Phi}^{d}\in\mathcal{S}(\mathcal{H}_{\hat{A}\hat{B}})$, and
$\Phi^{d}\in\mathcal{S}(\mathcal{H}_{\hat{A}\hat{B}})$, the last defined through
\eqref{eq:max-ent-def}. Suppose that $F(\widetilde{\Phi}^{d},\Phi^{d}%
)\geq1-\varepsilon$. Then, for all $\alpha>1$,%
\begin{equation*}
\log_{2}d\leq\widetilde{R}_{\alpha}(\widetilde{\Phi}^{d})+\frac{\alpha}%
{\alpha-1}\log_{2}\!\left(  \frac{1}{1-\varepsilon}\right)  .
\end{equation*}

\end{lemma}

\begin{proof}
We provide two different proofs of this statement. The first proof is similar
to that given in ref.~\cite[Proposition~4]{TWW14}, and we provide it for completeness. Let a measurement channel $\mathcal{M}\in\operatorname{CPTP}%
(\mathcal{H}_{\hat{A}\hat{B}},\mathcal{H}_{X})$ be defined as%
\begin{equation*}
\mathcal{M}(\cdot)\coloneqq\operatorname{Tr}[\Phi^{d}(\cdot)]|1\rangle
\!\langle1|+\operatorname{Tr}[(I-\Phi^{d})(\cdot)]|0\rangle\!\langle0|.
\end{equation*}
Recall that $F(\widetilde{\Phi}^{d},\Phi^{d})=\operatorname{Tr}[\Phi
^{d}\widetilde{\Phi}^{d}]$ because $\Phi^{d}$ is pure. Then, for $\sigma
\in\operatorname{PPT}^{\prime}(\mathcal{H}_{\hat{A}\hat{B}})$ and $\alpha>1$,
we find that
\begin{align}
&  \widetilde{D}_{\alpha}(\widetilde{\Phi}^{d}\Vert\sigma)\nonumber\\
&  \geq\widetilde{D}_{\alpha}(\mathcal{M}(\widetilde{\Phi}^{d})\Vert
\mathcal{M}(\sigma))\nonumber\\
&  =\frac{1}{\alpha-1}\log_{2}\!\left[
\begin{array}
[c]{c}%
\left(  \operatorname{Tr}[\Phi^{d}\widetilde{\Phi}^{d}]\right)  ^{\alpha
}\left(  \operatorname{Tr}[\Phi^{d}\sigma]\right)  ^{1-\alpha}+\\
\left(  1-\operatorname{Tr}[\Phi^{d}\widetilde{\Phi}^{d}]\right)  ^{\alpha
}\left(  \operatorname{Tr}[\sigma]-\operatorname{Tr}[\Phi^{d}\sigma]\right)
^{1-\alpha}%
\end{array}
\right] \nonumber\\
&  \geq\frac{1}{\alpha-1}\log_{2}\!\left[  \left(  \operatorname{Tr}[\Phi
^{d}\widetilde{\Phi}^{d}]\right)  ^{\alpha}\left(  \operatorname{Tr}[\Phi
^{d}\sigma]\right)  ^{1-\alpha}\right] \nonumber\\
&  \geq\frac{1}{\alpha-1}\log_{2}\!\left[  \left(  1-\varepsilon\right)
^{\alpha}\left(  \frac{1}{d}\right)  ^{1-\alpha}\right] \nonumber\\
&  =-\frac{\alpha}{\alpha-1}\log_{2}\!\left(  \frac{1}{1-\varepsilon}\right)
+\log_{2}d. \label{eq:ineq-chain-Rains-alpha-bnd}%
\end{align}
The first inequality follows from the data-processing inequality for
$\widetilde{D}_{\alpha}$ for $\alpha>1$. The second
inequality follows by dropping the second term inside the logarithm. The third
inequality follows by applying the assumption that $F(\widetilde{\Phi}%
^{d},\Phi^{d})\geq1-\varepsilon$, as well as the known fact that
\cite[Lemma~2]{R99}%
\begin{equation}
\operatorname{Tr}[\Phi^{d}\sigma]\leq\frac{1}{d}
\label{eq:ppt'-bound-ent-test}%
\end{equation}
for $\sigma\in\operatorname{PPT}^{\prime}(\mathcal{H}_{\hat{A}\hat{B}})$.
Indeed, observing that the set $\operatorname{PPT}^{\prime}(\mathcal{H}%
_{\hat{A}\hat{B}})$ is the same regardless of the basis in which the partial
transpose $\operatorname{id}\otimes T$ is taken, we can take the basis to be
the same as that for the maximally entangled state $\Phi^{d}$, and we find
that%
\begin{align}
\operatorname{Tr}[\Phi^{d}\sigma]  &  =\operatorname{Tr}[\Phi^{d}\left(
(\operatorname{id}\otimes T)(\operatorname{id}\otimes T)\right)
\sigma]\nonumber\\
&  =\operatorname{Tr}[(\operatorname{id}\otimes T)(\Phi^{d})(\operatorname{id}%
\otimes T)(\sigma)]\nonumber\\
&  =\frac{1}{d}\operatorname{Tr}[F^{d}(\operatorname{id}\otimes T)(\sigma
)]\nonumber\\
&  \leq\frac{1}{d}\left\Vert (\operatorname{id}\otimes T)(\sigma)\right\Vert
_{1}\leq\frac{1}{d}.
\end{align}
In the above, we used the facts that the partial transpose is self-inverse,
self-adjoint with respect to the Hilbert--Schmidt inner product, that
$(\operatorname{id}\otimes T)(\Phi^{d})=\frac{1}{d}F^{d}$, where $F^{d}%
\in\mathcal{L}(\mathcal{H}_{\hat{A}\hat{B}})$ is the unitary swap operator,
and the variational characterization of the trace norm of $X\in\mathcal{L}%
(\mathcal{H})$ as $\left\Vert X\right\Vert _{1}=\sup_{U}\left\vert
\operatorname{Tr}[XU]\right\vert $, with the supremum over every unitary
$U\in\mathcal{L}(\mathcal{H})$.

Returning to \eqref{eq:ineq-chain-Rains-alpha-bnd}, we have thus proven the
following inequality for all $\sigma\in\operatorname{PPT}^{\prime}%
(\mathcal{H}_{\hat{A}\hat{B}})$ and $\alpha>1$:%
\begin{equation*}
\widetilde{D}_{\alpha}(\widetilde{\Phi}^{d}\Vert\sigma)\geq-\frac{\alpha
}{\alpha-1}\log_{2}\!\left(  \frac{1}{1-\varepsilon}\right)  +\log_{2}d.
\end{equation*}
We conclude 
after taking an infimum over $\sigma
\in\operatorname{PPT}^{\prime}(\mathcal{H}_{\hat{A}\hat{B}})$ and applying the
definition of $\widetilde{R}_{\alpha}(\widetilde{\Phi}^{d})$.

A second proof follows by applying
Lemma~\ref{lem:pseudo-cont-sandwiched-renyi}\ below:%
\begin{align}
\widetilde{R}_{\alpha}(\widetilde{\Phi}^{d})  &  \geq\widetilde{R}_{\beta
}(\Phi^{d})+\frac{\alpha}{\alpha-1}\log_{2}F(\widetilde{\Phi}^{d},\Phi
^{d})\nonumber\\
&  \geq\widetilde{R}_{\beta}(\Phi^{d})+\frac{\alpha}{\alpha-1}\log
_{2}(1-\varepsilon)\nonumber\\
&  \geq\widetilde{R}_{1/2}(\Phi^{d})+\frac{\alpha}{\alpha-1}\log
_{2}(1-\varepsilon),
\end{align}
where we have employed the monotonicity of the sandwiched R\'enyi relative entropy
with respect to $\alpha$  \cite{KW20book}. Then consider that%
\begin{align}
\widetilde{R}_{1/2}(\Phi^{d})  &  =\inf_{\sigma\in\operatorname{PPT}^{\prime
}(\mathcal{H}_{\hat{A}\hat{B}})}\left[  -\log_{2}F(\Phi^{d},\sigma)\right]
\nonumber\\
&  =-\log_{2}\sup_{\sigma\in\operatorname{PPT}^{\prime}(\mathcal{H}_{\hat
{A}\hat{B}})}F(\Phi^{d},\sigma)\nonumber\\
&  =-\log_{2}\sup_{\sigma\in\operatorname{PPT}^{\prime}(\mathcal{H}_{\hat
{A}\hat{B}})}\operatorname{Tr}[\Phi^{d}\sigma]\nonumber\\
&  \geq\log_{2}d,
\end{align}
where we again have applied \eqref{eq:ppt'-bound-ent-test}.
\end{proof}

\begin{lemma}[Pseudo-continuity
bound]
\label{lem:pseudo-cont-sandwiched-renyi}For states $\rho_{0},\rho_{1}%
\in\mathcal{S}(\mathcal{H}_{AB})$, $\beta\in(1/2,1)$, and $\alpha
=\beta/\left(  2\beta-1\right)  >1$, the following bound holds%
\begin{equation}
\widetilde{R}_{\alpha}(\rho_{0})-\widetilde{R}_{\beta}(\rho_{1})\geq
\frac{\alpha}{\alpha-1}\log_{2}F(\rho_{0},\rho_{1}).
\label{eq:rains-pseudo-cont}%
\end{equation}

\end{lemma}

\begin{proof}
Recall the following bound from ref.~\cite[Lemma~1]{Wang2019states}:%
\begin{equation*}
\widetilde{D}_{\alpha}(\rho_{0}\Vert\sigma)-\widetilde{D}_{\beta}(\rho
_{1}\Vert\sigma)\geq\frac{\alpha}{\alpha-1}\log_{2}F(\rho_{0},\rho_{1}),
\end{equation*}
which holds for $\rho_{0},\rho_{1}\in\mathcal{S}(\mathcal{H}_{AB})$,
$\sigma\in\mathcal{L}_{+}(\mathcal{H}_{AB})$, $\beta\in(1/2,1)$, and
$\alpha=\beta/\left(  2\beta-1\right)  >1$. (An inspection of the proof there
reveals that it holds more generally for $\sigma\in\mathcal{L}_{+}%
(\mathcal{H}_{AB})$, and not just for $\sigma\in\mathcal{S}(\mathcal{H}_{AB}%
)$.) Then, for $\sigma\in\operatorname{PPT}^{\prime}(\mathcal{H}_{AB})$, we
can rewrite this as%
\begin{align}
\widetilde{D}_{\alpha}(\rho_{0}\Vert\sigma)  &  \geq\widetilde{D}_{\beta}%
(\rho_{1}\Vert\sigma)+\frac{\alpha}{\alpha-1}\log_{2}F(\rho_{0},\rho
_{1})\nonumber\\
&  \geq\inf_{\sigma\in\operatorname{PPT}^{\prime}(\mathcal{H}_{AB})}%
\widetilde{D}_{\beta}(\rho_{1}\Vert\sigma)+\frac{\alpha}{\alpha-1}\log
_{2}F(\rho_{0},\rho_{1})\nonumber\\
&  =\widetilde{R}_{\beta}(\rho_{1})+\frac{\alpha}{\alpha-1}\log_{2}F(\rho
_{0},\rho_{1}).
\end{align}
Since this inequality holds for every $\sigma\in\operatorname{PPT}^{\prime
}(\mathcal{H}_{AB})$, we conclude that%
\begin{equation*}
\inf_{\sigma\in\operatorname{PPT}^{\prime}(\mathcal{H}_{AB})}\widetilde
{D}_{\alpha}(\rho_{0}\Vert\sigma)\geq\widetilde{R}_{\beta}(\rho_{1}%
)+\frac{\alpha}{\alpha-1}\log_{2}F(\rho_{0},\rho_{1}),
\end{equation*}
which is equivalent to the desired inequality in \eqref{eq:rains-pseudo-cont}.
\end{proof}

\section{Proof of Corollary~\ref{cor:reg-rains-bnd}}

\label{sec:reg-Rains-up-bnd}

We begin with the following lemma.

\begin{lemma}[Dimension bound]
\label{lem:dim-bnd-ed}For $\varepsilon\in(0,1)$ and $\rho\in\mathcal{S}%
(\mathcal{H}_{AB})$, the following bounds hold%
\begin{equation*}
0\leq E_{d}^{\varepsilon}(\rho)\leq\log_{2}D+\log_{2}\!\left(  \frac
{1}{1-\varepsilon}\right)  ,
\end{equation*}
where%
\begin{equation*}
D\coloneqq\min\left\{  \dim(\mathcal{H}_{A}),\dim(\mathcal{H}_{B})\right\}  .
\end{equation*}
\end{lemma}

\begin{proof}
The lower bound on $E_{d}^{\varepsilon}(\rho)$ is a direct consequence of its
definition. For the other bound, we employ
Theorem~\ref{thm:non-asymp-ED-upper-bnd} to conclude that the following holds
for all $\alpha>1$:%
\begin{align}
E_{d}^{\varepsilon}(\rho)  &  \leq\widetilde{R}_{\alpha}(\rho)+\frac{\alpha
}{\alpha-1}\log_{2}\!\left(  \frac{1}{1-\varepsilon}\right) \nonumber\\
&  \leq\widetilde{R}_{\alpha}(\Phi^{D})+\frac{\alpha}{\alpha-1}\log
_{2}\!\left(  \frac{1}{1-\varepsilon}\right) \nonumber\\
&  \leq\log_{2}D+\frac{\alpha}{\alpha-1}\log_{2}\!\left(  \frac{1}%
{1-\varepsilon}\right)  .
\end{align}
The second inequality follows from a teleportation argument:\ $\rho$ can be realized from $\Phi^{D}$ by means of LOCC\ via the
teleportation protocol (see, e.g., ref.~\cite{W17}). The final inequality is a consequence of the same
reasoning given for Eq.~(5.2.88)\ of ref.~\cite{KW20book}. Since the bound holds
for all $\alpha>1$, we can take the limit $\alpha\rightarrow\infty$ to arrive
at the desired upper bound on $E_{d}^{\varepsilon}(\rho)$.
\end{proof}

We now state the following theorem.

\begin{theorem}[Regularized quantities]
\label{thm:ed-reg-form}For all $\varepsilon\in(0,1)$, $\rho\in\mathcal{S}%
(\mathcal{H}_{AB})$, and $\ell\in\mathbb{N}$, the following equalities hold%
\begin{align}
\liminf_{n\rightarrow\infty}\frac{1}{n}E_{d}^{\varepsilon}(\rho^{\otimes n})
&  =\frac{1}{\ell}\liminf_{n\rightarrow\infty}\frac{1}{n}E_{d}^{\varepsilon
}((\rho^{\otimes\ell})^{\otimes n}),\label{eq:lim-inf-ED-reg}\\
\limsup_{n\rightarrow\infty}\frac{1}{n}E_{d}^{\varepsilon}(\rho^{\otimes n})
&  =\frac{1}{\ell}\limsup_{n\rightarrow\infty}\frac{1}{n}E_{d}^{\varepsilon
}((\rho^{\otimes\ell})^{\otimes n}). \label{eq:lim-sup-ED-reg}%
\end{align}
Consequently, for all $\ell\in\mathbb{N}$,%
\begin{align}
E_{d}(\rho)  &  =\frac{1}{\ell}E_{d}(\rho^{\otimes\ell}),\label{eq:PADE-reg}\\
\widetilde{E}_{d}(\rho)  &  =\frac{1}{\ell}\widetilde{E}_{d}(\rho^{\otimes
\ell}). \label{eq:sc-PADE-reg}%
\end{align}

\end{theorem}

\begin{proof}
First consider that%
\begin{align}
\liminf_{n\rightarrow\infty}\frac{1}{n}E_{d}^{\varepsilon}(\rho^{\otimes n})
&  \leq\frac{1}{\ell}\liminf_{n\rightarrow\infty}\frac{1}{n}E_{d}%
^{\varepsilon}(\rho^{\otimes n\ell}),\\
\limsup_{n\rightarrow\infty}\frac{1}{n}E_{d}^{\varepsilon}(\rho^{\otimes n})
&  \geq\frac{1}{\ell}\limsup_{n\rightarrow\infty}\frac{1}{n}E_{d}%
^{\varepsilon}(\rho^{\otimes n\ell}),
\end{align}
because $\{\frac{1}{n\ell}E_{d}^{\varepsilon}(\rho^{\otimes n\ell}%
)\}_{n\in\mathbb{N}}$ is a subsequence of $\{\frac{1}{n}E_{d}^{\varepsilon
}(\rho^{\otimes n})\}_{n\in\mathbb{N}}$.

So it remains to prove the opposite inequalities. We first show that%
\begin{equation}
\limsup_{n\rightarrow\infty}\frac{1}{n}E_{d}^{\varepsilon}(\rho^{\otimes
n})\leq\frac{1}{\ell}\limsup_{n\rightarrow\infty}\frac{1}{n}E_{d}%
^{\varepsilon}(\rho^{\otimes n\ell}). \label{eq:tougher-up-bnd-ed}%
\end{equation}
To this end, we make use of an idea from the proof of Theorem~8 in ref.~\cite{TWW14}. For a fixed positive integer $\ell$, consider that%
\begin{multline}
\limsup_{n\rightarrow\infty}\frac{1}{n}E_{d}^{\varepsilon}(\rho^{\otimes
n})=\\
\max_{j\in\left\{  0,1,\ldots,\ell-1\right\}  }\limsup_{k\rightarrow\infty
}\frac{1}{k\ell+j}E_{d}^{\varepsilon}(\rho^{\otimes\left(  k\ell+j\right)  }).
\label{eq:lim-sup-rewrite}%
\end{multline}
Consider now a fixed $j\in\left\{  0,1,\ldots,\ell-1\right\}  $. Then we find
that%
\begin{align}
&  \frac{1}{k\ell+j}E_{d}^{\varepsilon}(\rho^{\otimes\left(  k\ell+j\right)
})\nonumber\\
&  \leq\frac{1}{k\ell+j}E_{d}^{\varepsilon}(\rho^{\otimes\left(  \left(
k+1\right)  \ell\right)  })\nonumber\\
&  =\frac{1}{\left(  k+1\right)  \ell}E_{d}^{\varepsilon}(\rho^{\otimes\left(
\left(  k+1\right)  \ell\right)  })\nonumber\\
&  \qquad+\frac{\ell-j}{\left(  k+1\right)  \ell\left(  k\ell+j\right)  }%
E_{d}^{\varepsilon}(\rho^{\otimes\left(  \left(  k+1\right)  \ell\right)  }).
\end{align}
We conclude that%
\begin{equation*}
\limsup_{k\rightarrow\infty}\frac{\ell-j}{\left(  k+1\right)  \ell\left(
k\ell+j\right)  }E_{d}^{\varepsilon}(\rho^{\otimes\left(  \left(  k+1\right)
\ell\right)  })=0
\end{equation*}
because, from Lemma~\ref{lem:dim-bnd-ed},%
\begin{equation*}
0\leq\frac{\ell-j}{\left(  k+1\right)  \ell\left(  k\ell+j\right)  }%
E_{d}^{\varepsilon}(\rho^{\otimes\left(  \left(  k+1\right)  \ell\right)  }),
\end{equation*}
and%
\begin{align*}
&  \frac{\ell-j}{\left(  k+1\right)  \ell\left(  k\ell+j\right)  }%
E_{d}^{\varepsilon}(\rho^{\otimes\left(  \left(  k+1\right)  \ell\right)
})\nonumber\\
&  \leq\frac{\ell-j}{\left(  k+1\right)  \ell\left(  k\ell+j\right)  }\left[
\left(  \left(  k+1\right)  \ell\right)  \log_{2}D+\log_{2}\!\left(  \frac
{1}{1-\varepsilon}\right)  \right] \nonumber\\
&  =\frac{\ell-j}{k\ell+j}\left[  \log_{2}D+\frac{1}{\left(  k+1\right)  \ell
}\log_{2}\!\left(  \frac{1}{1-\varepsilon}\right)  \right]  ,
\end{align*}
so that%
\begin{align*}
&  \limsup_{k\rightarrow\infty}\frac{\ell-j}{\left(  k+1\right)  \ell\left(
k\ell+j\right)  }E_{d}^{\varepsilon}(\rho^{\otimes\left(  \left(  k+1\right)
\ell\right)  })\nonumber\\
&  \leq\limsup_{k\rightarrow\infty}\frac{\ell-j}{k\ell+j}\left[  \log
_{2}D+\frac{1}{\left(  k+1\right)  \ell}\log_{2}\!\left(  \frac{1}%
{1-\varepsilon}\right)  \right] \nonumber\\
&  =0.
\end{align*}
Thus, we find that%
\begin{align}
&  \limsup_{k\rightarrow\infty}\frac{1}{k\ell+j}E_{d}^{\varepsilon}%
(\rho^{\otimes\left(  k\ell+j\right)  })\nonumber\\
&  \leq\limsup_{k\rightarrow\infty}\frac{1}{\left(  k+1\right)  \ell}%
E_{d}^{\varepsilon}(\rho^{\otimes\left(  \left(  k+1\right)  \ell\right)
})\nonumber\\
&  =\frac{1}{\ell}\limsup_{k\rightarrow\infty}\frac{1}{k+1}E_{d}^{\varepsilon
}(\rho^{\otimes\left(  \left(  k+1\right)  \ell\right)  }).
\end{align}
Since this upper bound is independent of $j$, we combine with
\eqref{eq:lim-sup-rewrite} to arrive at the desired inequality in
\eqref{eq:tougher-up-bnd-ed}. The inequality
\begin{equation*}
\liminf_{n\rightarrow\infty}\frac{1}{n}E_{d}^{\varepsilon}(\rho^{\otimes
n})\geq\frac{1}{\ell}\liminf_{n\rightarrow\infty}\frac{1}{n}E_{d}%
^{\varepsilon}(\rho^{\otimes n\ell})
\end{equation*}
can be proved in a similar fashion, by using the fact that%
\begin{multline}
\liminf_{n\rightarrow\infty}\frac{1}{n}E_{d}^{\varepsilon}(\rho^{\otimes
n})=\\
\min_{j\in\left\{  0,1,\ldots,\ell-1\right\}  }\liminf_{k\rightarrow\infty
}\frac{1}{k\ell+j}E_{d}^{\varepsilon}(\rho^{\otimes\left(  k\ell+j\right)  }).
\end{multline}
Thus we conclude the desired equalities in \eqref{eq:lim-inf-ED-reg}--\eqref{eq:lim-sup-ED-reg}.
Since the equalities in \eqref{eq:lim-inf-ED-reg}--\eqref{eq:lim-sup-ED-reg}
hold for all $\varepsilon\in\left(  0,1\right)  $, we conclude
\eqref{eq:PADE-reg}--\eqref{eq:sc-PADE-reg}\ after applying the definitions in
\eqref{eq:asymp-PADE}--\eqref{eq:strong-conv-PADE}, respectively.
\end{proof}

\begin{proof}
[Proof of Corollary~\ref{cor:reg-rains-bnd}]Apply
\eqref{eq:sc-PADE-reg} and then Theorem~\ref{thm:strong-conv-Rains}%
\ to the state $\rho^{\otimes\ell}$ to conclude that%
\begin{equation*}
\widetilde{E}_{d}(\rho)=\frac{1}{\ell}\widetilde{E}_{d}(\rho^{\otimes\ell
})\leq\frac{1}{\ell}R(\rho^{\otimes\ell}).
\end{equation*}
Since this upper bound holds for every $\ell\in\mathbb{N}$, the infimum over
$\ell\in\mathbb{N}$ is an upper bound as well.
\end{proof}

\section{Weak Converse Bound on PADE}

\label{sec:weak-conv-bnd}

In this appendix, we establish a bound on the non-asymptotic PADE, which can
be used to arrive at a weak converse bound on the asymptotic PADE.

\subsection{Upper bound on non-asymptotic PADE}

\begin{theorem}[Single-shot weak converse]
\label{thm:weak-converse-1-shot} For $\varepsilon\in\lbrack0,1/2]$ and
$\rho\in\mathcal{S}(\mathcal{H}_{AB})$, the following upper bound holds,%
\begin{equation}
E_{d}^{\varepsilon}(\rho)\leq\frac{1}{1-\varepsilon}\left[  R(\rho
)+h_{2}(\varepsilon)\right]  . \label{eq:ed-up-bnd-weak}%
\end{equation}

\end{theorem}

\begin{proof}
Let $\mathcal{L}\in\operatorname{LOCC}(\mathcal{H}_{AB},\mathcal{H}_{X\hat
{A}\hat{B}})$ and $p\in\left[  0,1\right]  $ be the elements of an arbitrary
non-asymptotic PADE\ protocol. Consider that%
\begin{align}
R(\rho)  &  \geq pR(\widetilde{\Phi}^{d})+\left(  1-p\right)  R(\sigma
)\nonumber\\
&  \geq pR(\widetilde{\Phi}^{d})\nonumber\\
&  \geq p\left[  \left(  1-\varepsilon\right)  \log_{2}d-h_{2}(\varepsilon
)\right] \nonumber\\
&  \geq p\left(  1-\varepsilon\right)  \log_{2}d-h_{2}(\varepsilon),
\end{align}
where the first inequality follows from applying Theorem~\ref{thm:main} and
Corollary~\ref{rem:selective-PPT-actual-rel-ents}\ to $\mathcal{L}$ and $\rho
$. The second inequality follows because $R(\sigma)\geq0$, and the third
inequality follows from applying
Lemma~\ref{lem:rains-ineq-approx-max-ent-weak-conv} below. Rewriting the
inequality $R(\rho)\geq p\left(  1-\varepsilon\right)  \log_{2}d-h_{2}%
(\varepsilon)$, we find that
\begin{equation*}
p\log_{2}d\leq\frac{1}{1-\varepsilon}\left[  R(\rho)+h_{2}(\varepsilon
)\right]  .
\end{equation*}
Since the upper bound depends only on $\rho$ and $\varepsilon$ and it holds
for all $\mathcal{L}\in\operatorname{LOCC}(\mathcal{H}_{AB},\mathcal{H}%
_{X\hat{A}\hat{B}})$ and $p\in\left[  0,1\right]  $, we conclude
\eqref{eq:ed-up-bnd-weak}\ after applying the definition of $E_{d}%
^{\varepsilon}(\rho)$.
\end{proof}

\begin{lemma}[Weak convexity]
\label{lem:rains-ineq-approx-max-ent-weak-conv}Let $\varepsilon\in\left[
0,1/2\right]  $, $\widetilde{\Phi}^{d}\in\mathcal{S}(\mathcal{H}_{\hat{A}%
\hat{B}})$, and $\Phi^{d}\in\mathcal{S}(\mathcal{H}_{\hat{A}\hat{B}})$ as
defined in \eqref{eq:max-ent-def}. Suppose that $F(\widetilde{\Phi}^{d}%
,\Phi^{d})\geq1-\varepsilon$. Then%
\begin{equation*}
\left(  1-\varepsilon\right)  \log_{2}d\leq R(\widetilde{\Phi}^{d}%
)+h_{2}(\varepsilon).
\end{equation*}

\end{lemma}

\begin{proof}
Let a measurement channel $\mathcal{M}\in\operatorname{CPTP}(\mathcal{H}%
_{\hat{A}\hat{B}},\mathcal{H}_{X})$ be defined as%
\begin{equation*}
\mathcal{M}(\cdot)\coloneqq\operatorname{Tr}[\Phi^{d}(\cdot)]|1\rangle
\!\langle1|+\operatorname{Tr}[(I-\Phi^{d})(\cdot)]|0\rangle\!\langle0|.
\end{equation*}
Recall that $F(\widetilde{\Phi}^{d},\Phi^{d})=\operatorname{Tr}[\Phi
^{d}\widetilde{\Phi}^{d}]$ because $\Phi^{d}$ is pure. Then, for $\sigma
\in\operatorname{PPT}^{\prime}(\mathcal{H}_{\hat{A}\hat{B}})$, we find that
\begin{align}
D(\widetilde{\Phi}^{d}\Vert\sigma)  &  \geq D(\mathcal{M}(\widetilde{\Phi}%
^{d})\Vert\mathcal{M}(\sigma))\nonumber\\
&  =\operatorname{Tr}[\Phi^{d}\widetilde{\Phi}^{d}]\log_{2}\!\left(
\frac{\operatorname{Tr}[\Phi^{d}\widetilde{\Phi}^{d}]}{\operatorname{Tr}%
[\Phi^{d}\sigma]}\right) \nonumber\\
&  \qquad+\left(  1-\operatorname{Tr}[\Phi^{d}\widetilde{\Phi}^{d}]\right)
\log_{2}\!\left(  \frac{1-\operatorname{Tr}[\Phi^{d}\widetilde{\Phi}^{d}%
]}{\operatorname{Tr}[\sigma]-\operatorname{Tr}[\Phi^{d}\sigma]}\right)
\nonumber\\
&  =-h_{2}(1-\operatorname{Tr}[\Phi^{d}\widetilde{\Phi}^{d}%
])-\operatorname{Tr}[\Phi^{d}\widetilde{\Phi}^{d}]\log_{2}(\operatorname{Tr}%
[\Phi^{d}\sigma])\nonumber\\
&  \qquad-\left(  1-\operatorname{Tr}[\Phi^{d}\widetilde{\Phi}^{d}]\right)
\log_{2}(\operatorname{Tr}[\sigma]-\operatorname{Tr}[\Phi^{d}\sigma
])\nonumber\\
&  \geq-h_{2}(\varepsilon)+(1-\varepsilon)
\log_{2}d.
\end{align}
The first inequality follows from the data-processing inequality for $D$. The second inequality follows because%
\begin{equation*}
-\left(  1-\operatorname{Tr}[\Phi^{d}\widetilde{\Phi}^{d}]\right)  \log
_{2}(\operatorname{Tr}[\sigma]-\operatorname{Tr}[\Phi^{d}\sigma])\geq0.
\end{equation*}
Also, we have applied the assumption that $1-\operatorname{Tr}[\Phi^{d}%
\widetilde{\Phi}^{d}]\leq\varepsilon$ and that the binary entropy function is
monotone increasing on the interval $\left[  0,1/2\right]  $. Additionally, we
applied \eqref{eq:ppt'-bound-ent-test}. Since the inequality $D(\widetilde
{\Phi}^{d}\Vert\sigma)\geq-h_{2}(\varepsilon)+(1-\varepsilon)\log_{2}d$ holds
for all $\sigma\in\operatorname{PPT}^{\prime}(\mathcal{H}_{\hat{A}\hat{B}})$,
we conclude the desired inequality after taking the infimum over $\sigma
\in\operatorname{PPT}^{\prime}(\mathcal{H}_{\hat{A}\hat{B}})$.
\end{proof}

\subsection{Weak-converse upper bound on asymptotic PADE}

Here we show that the asymptotic PADE is bounded from above by the Rains
relative entropy as
\begin{equation*}
E_{d}(\rho)\leq R(\rho).
\end{equation*}
This inequality actually follows directly from
Theorem~\ref{thm:strong-conv-Rains} and the fact that $E_{d}(\rho)
\leq\widetilde{E}_{d}(\rho)$, but here we see it as a direct consequence of
Theorem~\ref{thm:weak-converse-1-shot}. Indeed, consider that
\begin{align}
E_{d}(\rho)  &  = \inf_{\varepsilon\in(0,1)} \liminf_{n \to\infty}\frac{1}{n}
E_{d}^{\varepsilon}(\rho^{\otimes n})\nonumber\\
&  \leq\inf_{\varepsilon\in(0,1)} \liminf_{n \to\infty} \left(  \frac
{1}{n(1-\varepsilon)}\left[  R(\rho^{\otimes n} )+h_{2}(\varepsilon)\right]
\right) \nonumber\\
&  \leq\inf_{\varepsilon\in(0,1)} \liminf_{n \to\infty} \left(  \frac
{1}{1-\varepsilon}\left[  R(\rho)+\frac{h_{2}(\varepsilon)}{n}\right]  \right)
\nonumber\\
&  = \inf_{\varepsilon\in(0,1)} \frac{1}{1-\varepsilon} R(\rho)\nonumber\\
&  = R(\rho).
\end{align}
For the first inequality, we
have applied Theorem~\ref{thm:weak-converse-1-shot}
and for the second, we
have made use of the subadditivity of the Rains relative entropy.
Employing the same kind of argument from Appendix~\ref{sec:reg-Rains-up-bnd},
we arrive at 
\begin{equation*}
E_{d}(\rho) \leq\inf_{\ell\in\mathbb{N} } R(\rho^{\otimes\ell})
\end{equation*}
as a regularized Rains relative entropy bound.

\end{document}